\documentclass[final]{IEEEtran}
\usepackage{cite}
\usepackage{color}
\usepackage{threeparttable}
\usepackage{graphicx}
\usepackage{picinpar}
\usepackage[cmex10]{amsmath}
\usepackage{amsmath,amsfonts,amssymb}
\usepackage{subfigure}
\usepackage{changepage}
\usepackage{algorithm}
\usepackage{algorithmic}
\usepackage{stfloats}
\usepackage{bm}
\usepackage{enumerate}
\newtheorem{proof}{Proof}

\begin{document}
\newtheorem{lemma}{Lemma}
\newtheorem{corol}{Corollary}
\newtheorem{theorem}{Theorem}
\newtheorem{proposition}{Proposition}
\newtheorem{definition}{Definition}
\newcommand{\e}{\begin{equation}}
\newcommand{\ee}{\end{equation}}
\newcommand{\eqn}{\begin{eqnarray}}
\newcommand{\eeqn}{\end{eqnarray}}
\newenvironment{shrinkeq}[1]
{ \bgroup
\addtolength\abovedisplayshortskip{#1}
\addtolength\abovedisplayskip{#1}
\addtolength\belowdisplayshortskip{#1}
\addtolength\belowdisplayskip{#1}}
{\egroup\ignorespacesafterend}
\title{\LARGE Over-Sampling Codebook-Based Hybrid Minimum Sum-Mean-Square-Error Precoding for Millimeter-Wave 3D-MIMO}

\author{Jiening Mao, Zhen Gao, Yongpeng Wu, and Mohamed-Slim Alouini 

\thanks{J. Mao and Z. Gao are with both Advanced Research Institute of Multidisciplinary Science (ARIMS) and School of Information and Electronics, Beijing Institute of Technology, Beijing 100081, China (E-mail: gaozhen16@bit.edu.cn).}
\thanks{Y. Wu is with Department aof Electronic Engineering, Shanghai Jiao Tong University, Shanghai 200240, China.}
\thanks{M. -S. Alouini is with
King Abdullah University of Science and Technology (KAUST), Thuwal, Saudi Arabia.}
}

\maketitle
\begin{abstract}

Hybrid precoding design is challenging for millimeter-wave (mmWave) massive MIMO.
Most prior hybrid precoding schemes are designed to maximize the sum spectral efficiency (SSE), while seldom investigate the bit-error-rate (BER).
Therefore, this letter designs an over-sampling codebook (OSC)-based hybrid minimum sum-mean-square-error (min-SMSE) precoding to optimize the BER,
where multi-user transmission with multiple data streams for each user is considered. 
Specifically, given the effective baseband channel consisting of the real channel and analog precoding, we first design the digital precoder/combiner based on min-SMSE criterion to optimize the BER. To further reduce the SMSE between the transmit and receive signals, we propose an OSC-based joint analog precoder/combiner (JAPC) design.
Simulation results show that the proposed scheme can achieve the better performance than its conventional counterparts.

\end{abstract}
 \vspace*{-1mm}
 \vspace*{-2.0mm}
\IEEEpeerreviewmaketitle
 \vspace*{-2.0mm}
\section{Introduction}
 \vspace*{-1.0mm}
Millimeter-wave (mmWave) massive MIMO with hybrid precoding is a promising key technique for 5G \cite{{2},{R1}}. However, the hybrid precoding design can be challenging due to the non-convex constraints introduced by the analog precoding with modulus constraint \cite{{8}}.

Hybrid analog/digital precoding has been investigated in \cite{3,4,5,6,8,R2}. Specifically,
a compressive sensing based hybrid precoding and an energy efficient hybrid precoding have been proposed in \cite{8} and \cite{{R2}}, respectively, while they are limited to single-user MIMO systems.
A phased zero forcing hybrid precoding has been proposed for multi-user mmWave MIMO \cite{{3}}, while only single-antenna user is considered.
Moreover, to support multi-user mmWave MIMO with each user equipped with multiple antennas and single RF chain,
a two-stage multi-user hybrid precoding and a heuristic hybrid precoding have been proposed in \cite{{4}} and \cite{{5}}, respectively.
However, how to effectively support multi-user and multi-stream transmission for each user is not well addressed.
Recently, a hybrid block diagonalization (BD) precoding is developed from the traditional BD precoding in \cite{{6}}, which can be used to support multi-user massive MIMO with each user equipped with multiple antennas and multiple RF chains.
On the other hand, most prior work \cite{{3},{4},{5},{6},{8},{R2}} have designed the analog precoding to maximize the sum spectral efficiency (SSE) other than to optimize the bit-error-rate (BER) performance, and they usually consider the ideal uniform linear array (ULA) other than the practical uniform planar array (UPA).

In this letter, we propose an over-sampling codebook (OSC)-based
hybrid minimum sum-mean-square-error (min-SMSE) precoding scheme for mmWave multi-user three-dimensional (3D)-MIMO systems to optimize the BER.
Specifically, given the effective baseband channel consisting of the real channel and analog precoding, we first design the digital precoder/combiner to minimize the SMSE between the transmit symbols and receive symbols.
Furthermore, we propose an OSC-based joint analog precoder/combiner (JAPC) design by further reducing the SMSE for improved BER.
Simulation results confirm that the proposed scheme has the better performance than the conventional schemes.
To the best of our knowledge, this is the first work to jointly design digital and analog precoder/combiner for mmWave massive multiuser MIMO based on minimum SMSE criterion.


\textit{Notations}: 
Lower-case and upper-case boldface letters denote vectors and matrices, respectively;
${\left(  \cdot  \right)\!^T}$, ${\left( \cdot \right)\!^H}$, and ${\mathop{\rm tr}\nolimits} \left(  \cdot  \right)$ denote the transpose, conjugate transpose, and trace of a matrix, respectively; ${\left(  \cdot  \right)^{(i)}}$ and ${\left(  \cdot  \right)\!_{i,j}}$ represent the ${i}$-th column and ${i}$-th row and ${j}$-th column element of a matrix, respectively; ${\bf{I}}_N$ represents the ${N \!\times \!N}$ identity matrix; ${\left\| {\cdot} \right\|_2}$, ${\left\| {\cdot} \right\|_{\rm F}}$, ${\mathop{\rm diag}\nolimits} \left(  \cdot  \right)$, ${\mathop{\rm angle}\nolimits} \left(  \cdot  \right)$, and ${{\mathbb{E}}\left[  \cdot  \right]}$ represent the ${\ell _2}$-nrom, Frobenius norm, diagonalization, phase extraction, and expectation operators, respectively.

\vspace*{-2mm}
\section{System Model}
We consider the downlink mmWave multi-user 3D-MIMO system with the hybrid precoding.
Assume that
the base station (BS) employs ${N_t}$ antennas but ${M_t} \ll {N_t}$ RF chains to simultaneously serve ${K}$ users, and each user employs $N_r$ antennas but ${M_r} \ll {N_r}$ RF chains to support ${N_s} \le {M_r}$ data streams.
Therefore, the BS can simultaneously support $KN_s\le M_t$ data streams.
In the downlink transmission with $KM_r = M_t$, the received signal at the $k$-th user can is
\begin{equation}
\label{equ:eq0}
{{\bf{y}}_k} = {\bf{V}}_k^H{\bf{M}}_k^H\left( {\gamma} {{{\bf{H}}_k}{\bf{FWx}} + {{\bf{n}}_k}} \right),
\end{equation}
where ${{\bf{V}}_k}\in {{\mathbb{C}}^{ {M_r}\times{N_s} }}$ and ${{\bf{M}}_k}\in {{\mathbb{C}}^{{N_r} \times {M_r}}}$ represent the digital combiner and analog combiner of the ${k}$-th user, respectively,
${{\bf{H}}_k} \in {{\mathbb{C}}^{{N_r} \times {N_t}}}$ is the downlink channel associated with the BS and the ${k}$-th user, ${\bf{F}} = [{{\bf{F}}_1},{{\bf{F}}_2}, \cdots ,{{\bf{F}}_K}]\in {{\mathbb{C}}^{ {N_t}\times{M_t} }}$ and ${\bf{W}} = [{{\bf{W}}_1},{{\bf{W}}_2}, \cdots ,{{\bf{W}}_K}]\in {{\mathbb{C}}^{ {M_t}\times{KN_s} }}$ represent
the analog precoder and digital precoder at the BS, respectively,
${{\bf{F}}_k}\in {{\mathbb{C}}^{{N_t} \times {M_r}}}$ and ${{\bf{W}}_k}\in {{\mathbb{C}}^{{M_t} \times {N_s}}}$ ${\left( {1 \le k \le K} \right)}$ represent the analog precoder and digital precoder for the ${k}$-th user, respectively.
${\gamma}$ is the normalization factor so that ${{\gamma ^2}}{\mathop{\rm tr}\nolimits}  \left({\bf{FW}}{{\bf{W}}^H}{{\bf{F}}^H} \right) = {P_t}$, and ${P_t}$ is the transmit power.
${\bf{x}} = {\left[ {{{\bf{x}}_1}^T,{{\bf{x}}_2}^T, \cdots ,{{\bf{x}}_K}^T} \right]^T}\in {{\mathbb{C}}^{ {KN_s}\times{1} }}$ represents  the transmitted signal vector for ${K}$ users.
${{\bf{n}}_k}\in {{\mathbb{C}}^{{N_r} \times 1}}$ is the additive white Gaussian noise (AWGN) vector at the $k$-th user.

The mmWave MIMO
channel matrix $\bf{H}$ is assumed to be a sum of the contributions of ${N_c}$ scattering clusters, and each cluster contains ${N_p}$ propagation paths \cite{{8}}, i.e.,
\begin{align}
\begin{array}{*{20}{l}}
\hspace{-2mm}
{{{\bf{H}}_k} = \sqrt {\dfrac{{{N_t}{N_r}}}{{{N_c}{N_p}}}} \sum\limits_{i = 1}^{{N_c}} {\sum\limits_{l = 1}^{{N_p}} {\alpha _{il}^k{\bf{a}}_r^k\left( {\theta _{il}^{rk},\phi _{il}^{rk}} \right){\bf{a}}_t^k{{\left( {\theta _{il}^{tk},\phi _{il}^{tk}} \right)}^H}} } },
\end{array}
\end{align}
where ${\alpha _{il}^k}$ is the complex gain of the ${i}$-th path in the ${l}$-th cluster and $\alpha _{il}^k \sim {\mathop{{\cal C}{\cal N}}\nolimits} \left(0,1\right)$.
${\theta _{il}^{rk}\left( {\phi _{il}^{rk}} \right)}$ and ${\theta _{il}^{tk}\left( {\phi _{il}^{tk}} \right)}$ are the azimuth (elevation) angles of arrival (AoAs) and departure (AoDs) of the ${i}$-th path in the ${l}$-th cluster, respectively.
The angles in each cluster follow the uniform distribution and have the constant angle spreads (standard deviation),
which can be  denoted by $\sigma _\phi ^{tk}$, $\sigma _\theta ^{tk}$, ${\sigma _\phi ^{rk}}$, and $\sigma _\theta ^{rk}$, respectively. ${{\bf{a}}_r^k\left( {\theta _{il}^{rk},\phi _{il}^{rk}} \right)}$ and ${{\bf{a}}_t^k\left( {\theta _{il}^{tk},\phi _{il}^{tk}} \right)}$ are the normalized receive and transmit array response vectors at an azimuth (elevation) angle ${\theta _{il}^{rk}\left( {\phi _{il}^{rk}} \right)}$ and ${\theta _{il}^{tk}\left( {\phi _{il}^{tk}} \right)}$, respectively.
\begin{figure}
     \centering
     \includegraphics[width=8.5cm, keepaspectratio]%
     {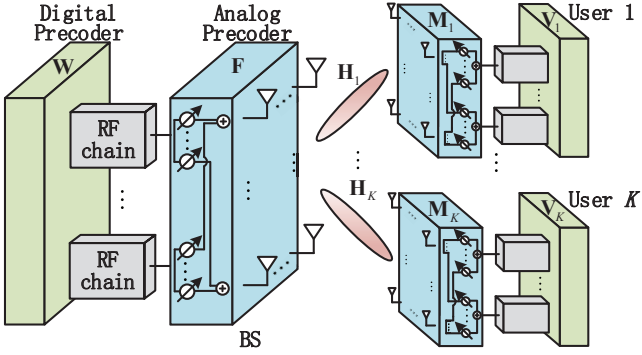}
     \vspace*{-1.5mm}
     \caption{The mmWave 3D-MIMO with hybrid precoding, where the BS can support multiple users with multi-stream transmission for each user.}
     \label{fig:1}
\end{figure}
 For UPA in the ${yz}$-plane with ${N_y}$ and ${N_z}$ elements on the ${y}$ and ${z}$ axes, respectively, the array response vector can be written as
\begin{align}
\label{equ:model1}
\hspace{-4.6mm}
~~{{{\bf{a}}_{{\rm{UPA}}}}\!\!\left( {\theta ,\phi } \right)} \!&= \!\frac{1}{{\sqrt {{N_y}{N_z}} }}\big[1, \cdots \!,{e^{j\frac{{2\pi }}{\lambda }d\left( {n\sin \left( \theta   \right)\!\cos \left( \phi \right) + m\sin \left( \phi  \right)} \right)}},\notag\\
 &\cdots ,{e^{j\frac{{2\pi }}{\lambda }d\left( {\left( {N_y - 1} \right)\sin \left( \theta \right)\cos \left( \phi   \right) + \left( {N_z - 1} \right)\sin \left( \phi   \right)} \right)}}\big]^T,
\end{align}
where ${1 \le n < N_y}$ and ${1 \le m < N_z}$ are the ${y}$ and ${z}$ indices of an antenna element, respectively. Finally, ${\lambda }$ and ${d}$ denote the wavelength and adjacent antenna spacing, respectively.

\section{OSC-Based Hybrid Min-SMSE Precoding}
In this section, we first propose a min-SMSE criterion based digital precoding design. To further improve the BER performance by reducing SMSE, we propose a JAPC design based on our proposed OSC in the analog part.

\subsection{Digital Precoding Design}
Assume that the ${k}$-th user's effective baseband channel ${\bf{H}}_{{\rm{eff}}}^k = {\bf{M}}_k^H{{\bf{H}}_k}{\bf{F}} \in {{\mathbb{C}}^{{M_r} \times {M_t}}}$ ${\left( {1 \le k \le K} \right)}$ is given, the estimated signal at the ${k}$-th user can be rewritten as
\begin{align}
{{{\bf{\hat x}}}_k} = {{\bf{y}}_k}/{\gamma} = {\bf{V}}_k^H{\bf{H}}_{{\rm{eff}}}^k{\bf{Wx}} + {\bf{V}}_k^H{\bf{M}}_k^H{{\bf{n}}_k}/{\gamma}.
\end{align}
Then the ${k}$-th user's the mean square error (MSE) can be expressed as ${\xi _k} = {\mathbb{E}}\left[ {\left\| {{{{\bf{\hat x}}}_k} - {{\bf{x}}_k}} \right\|_2^2} \right]$.
Moreover, since ${{\bf{x}}_k}$ and ${{\bf{n}}_k}$ are mutually independent,
$\xi _k$ can be further expressed as
\hspace{-4.6mm}
\begin{align}
\label{equ:eq2}
{\xi _k}\! &= \! {\rm{tr}} \! \left ({\bf{V}}_k^H{\bf{H}}_{{\rm{eff}}}^k{\bf{W}}{{\bf{W}}^H}\!{\left( {{\bf{H}}_{{\rm{eff}}}^k} \right)^H}{{\bf{V}}_k}\! \right)\!+ \!{\frac{{{\sigma ^2}}}{{{\gamma ^2}}}}{\rm{tr}}\!\left ({\bf{V}}_k^H{\bf{M}}_k^H{{\bf{M}}_k}\!{{\bf{V}}_k} \right) \notag\\
 &- {\rm{tr}}\left ({\bf{V}}_k^H{{\bf{H}}_{\rm eff}^k}{{\bf{W}}_k} \right) - {\rm{tr}}\left ({\bf{W}}_k^H{\left( {{\bf{H}}_{{\rm{eff}}}^k} \right)^H}{{\bf{V}}_k} \right)+ {N_s}.
\end{align}
In this way, the digital precoder/combiner design can be formulated as the following optimization problem
\begin{align}
\label{equ:digital1}
&\mathop {\min }\limits_{{\bf{W}},{\bf{V}}^H} \sum\nolimits_{k = 1}^K {{\xi _k}} \notag\\
&{\rm s.t.}\quad{{\gamma ^2}}{\mathop{\rm tr}\nolimits}  \left({\bf{FW}}{{\bf{W}}^H}{{\bf{F}}^H} \right) = {P_t},
\end{align}
where ${{\bf{V}}\!\!=\! {\mathop{\rm diag}\nolimits} \left({{{\bf{V}}_1},\!\!  {{\bf{V}}_2},\!  \cdots ,\!\!  {{\bf{V}}_K}} \right)}$ represents the digital combiner for $K$ users.
Moreover, since $K$ users will estimate their respective signals independently in the downlink, we can design the digital combiner for each user separately. Hence, the problem (\ref{equ:digital1}) can be decomposed as the subproblems ${\bf{V}}_k^H = \arg\mathop {\min\limits_{{\bf{V}}_k^H}{\xi _k}}$ for ${k = 1,2, \cdots ,K}$.
Consider the partial derivative of ${{\bf{V}}_k^H}$ with respect to ${\xi _k}$ and let it equal 0, we arrive at
\begin{align}
\label{equ:v}
\!\!\!\!\!\!{\bf{V}}_k^H\! \!  =\! \!  {\bf{W}}_k^H{\left( {{\bf{H}}_{\rm eff}^k} \right)^H}{\left ({{\bf{H}}_{\rm eff}^k}{\bf{W}}{{\bf{W}}^H}{\left( {{\bf{H}}_{\rm eff}^k} \right)^H}\! \!\! \!  +\!   {\frac{{{\sigma ^2}}}{{{\gamma ^2}}}}{\bf{M}}_k^H{{\bf{M}}_k} \right)^{\!\! - 1}}\!\!\!.
\end{align}

On the other hand, for the digital precoder ${\bf{W}}$ at the BS, the SMSE of ${K}$ users will be considered, i.e., $\xi  = \sum\nolimits_{k = 1}^K {{\xi _k}}  = {\mathop{\mathbb{E}}\nolimits} \left [\left\| {{\bf{\hat x}} - {\bf{x}}} \right\|_2^2 \right]$.
Similar to (\ref{equ:eq2}), $\xi$ can be expressed as
\begin{align}
\label{equ:eq3}
\xi  &= {\mathop{\rm tr}\nolimits} \left ({\bf{V}}^H{{\bf{H}}_{\rm eff}}{\bf{W}}{{\bf{W}}^H}{\bf{H}}_{\rm eff}^H{{\bf{V}}} \right) - {\mathop{\rm tr}\nolimits} \left ({\bf{V}}^H{{\bf{H}}_{\rm eff}}{\bf{W}} \right)\notag\\
 & - {\mathop{\rm tr}\nolimits} \left ({{\bf{W}}^H}{\bf{H}}_{\rm eff}^H{{\bf{V}}} \right) + {\frac{{{\sigma ^2}}}{{{\gamma ^2}}}}{\mathop{\rm tr}\nolimits} \left ({\bf{V}}^H{{\bf{M}}^H}{\bf{M}}{{\bf{V}}} \right) + K{N_s}.
\end{align}
Consider the partial derivative of ${\bf{W}}$ with respect to ${\xi}$ and let it equal 0, the optimal ${\bf{W}}$ can be
\begin{align}
\label{equ:w}
{\bf{W}} = {\left ({\bf{H}}_{\rm eff}^H{{\bf{V}}}{\bf{V}}^H{{\bf{H}}_{\rm eff}} \right)^{ - 1}}{\bf{H}}_{\rm eff}^H{{\bf{V}}}.
\end{align}
From (\ref{equ:v}) and (\ref{equ:w}), we can observe that the solutions of ${{\bf V}_k}$ and ${\bf W}$ depend on each other. To decouple these two solutions,
we first obtain an initial solution of V, denoted by ${{\bf{V}}_{{\rm{ini}}}} = {\rm{diag}}\left( {{\bf{V}}_{{\rm{ini}}}^1,{\bf{V}}_{{\rm{ini}}}^2, \cdots ,{\bf{V}}_{{\rm{ini}}}^K} \right)$, where each ${\bf{V}}_{{\rm{ini}}}^k \in {{\mathbb{C}}^{{N_s} \times {N_s}}}$ ${\left( {1 \le k \le K} \right)}$ can be an arbitrary unitary matrix. Moreover, we can obtain ${\bf W}$ by substituting the initial combiner ${\bf V_{\rm ini}}$ into (9), and then obtain ${\bf V}_k$ according to (7).
Note that different ${\bf V_{\rm ini}}$'s can lead to different digital precoder/combiner solutions \{${\bf W}$, ${\bf V}$\}, but these different solutions \{${\bf W}$, ${\bf V}$\} are equivalent. Specifically, we
first introduce the {\textbf {Lemma 1}} as follows.

\begin{figure*}[t]
     \centering
     \includegraphics[width=18cm, keepaspectratio]%
     {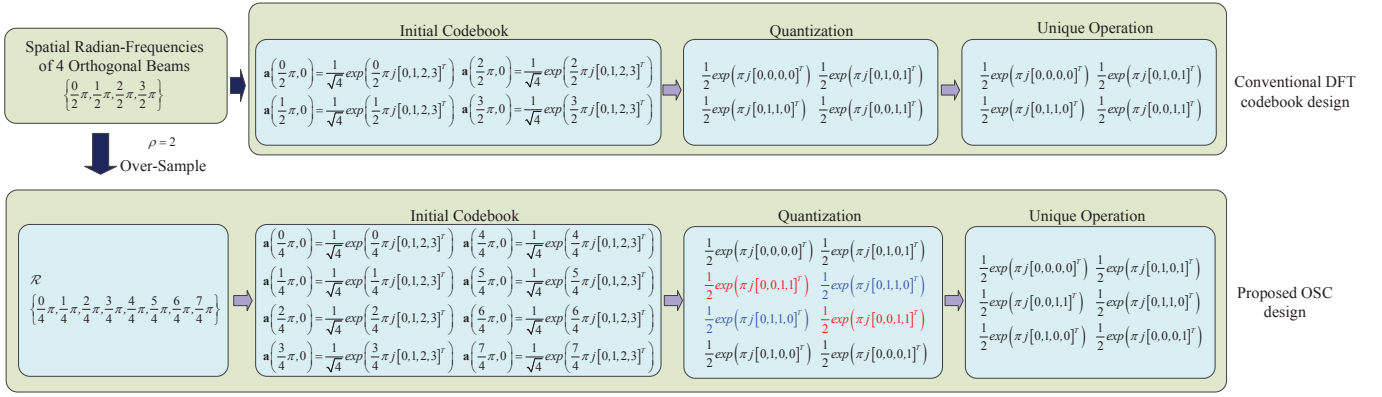}
     \caption{An example of the proposed OSC design, where a 4-element ULA with  $\rho = 2$ and $q = 1$ bit is considered.} 
     \label{fig:fig5}
\end{figure*}

\begin{lemma}
Assume that the matrix ${\bf{A}} = {\left[ {{\bf{A}}_1^T, \cdots ,{\bf{A}}_N^T} \right]^T}$ is invertible. Then
\begin{align}
\label{equ:lemma 1}
{{\bf{A}}_i}{\left( {{{\bf{A}}^H}{\bf{A}}} \right)^{ - 1}}{\bf{A}}_j^H \!=\! \delta \left( {i \!-\! j} \right){\bf{I}},1 \!\le\! i \!\le\! N,1 \!\le\! j \!\le\! N.
\end{align}
\end{lemma}
\begin{proof}
Due to the invertibility of ${\bf A}$, we can obtain
${\bf{A}}{\left( {{{\bf{A}}^H}{\bf{A}}} \right)^{ - 1}}{{\bf{A}}^H} = {\bf{A}}{{\bf{A}}^{ - 1}}{\left( {{{\bf{A}}^H}} \right)^{ - 1}}{{\bf{A}}^H} = {\bf{I}}$.
Meanwhile,
\begin{align}
\begin{array}{l}
{\bf{A}}{\left( {{{\bf{A}}^H}{\bf{A}}} \right)^{ - 1}}{{\bf{A}}^H} = {\rm{diag\big\{ }}{{\bf{A}}_1}{\left( {{{\bf{A}}^H}{\bf{A}}} \right)^{ - 1}}{\bf{A}}_1^H,\\
\qquad\qquad\qquad\quad\ \cdots ,{{\bf{A}}_N}{\left( {{{\bf{A}}^H}{\bf{A}}} \right)^{ - 1}}{\bf{A}}_N^H{\rm{\big\} }}.
\end{array}
\end{align}

Hence, (\ref{equ:lemma 1}) is proven, where $\delta(i-j)=1$ only if $j=i$,  otherwise $\delta(i-j)=0$.
\end{proof}
Given $\{{\bf V}_{\rm ini}^k\}_{k=1}^K$ being arbitrary unitary matrices, ${\bf{W}}\!\!=\!\!{\left( {{\bf{H}}_{{\rm{eff}}}^H{{\bf{H}}_{{\rm{eff}}}}} \right)^{ - 1}}{\bf{H}}_{{\rm{eff}}}^H{\bf{V}}_{\rm ini}$ and ${{\bf{W}}_k} \!\!=\!\!{\left( {{\bf{H}}_{{\rm{eff}}}^H{{\bf{H}}_{{\rm{eff}}}}} \right)^{ - 1}}{\left( {{\bf{H}}_{{\rm{eff}}}^k} \right)^H}{{\bf{V}}_{{\rm{ini}}}^k}$ according to (\ref{equ:w}). Then we have ${\bf{W}}{{\bf{W}}^H} = {\left( {{\bf{H}}_{{\rm{eff}}}^H{{\bf{H}}_{{\rm{eff}}}}} \right)^{ - 1}}$.
Consider the invertibility of matrix ${\bf H}_{\rm eff}$ and {\textbf {Lemma 1}},
\begin{align}
\label{equ:lemma 20}
{\bf{V}}_k^H
 & \mathop \approx \limits^a {\bf{W}}_k^H{\left( {{\bf{H}}_{{\rm{eff}}}^k} \right)^H}{\left( {{\bf{H}}_{{\rm{eff}}}^k{\bf{W}}{{\bf{W}}^H}{{\left( {{\bf{H}}_{{\rm{eff}}}^k} \right)}^H}} + {\frac{{{\sigma ^2}}}{{{\gamma ^2}}}}{\bf I}\right)^{ - 1}}\notag\\
 &= \frac{{{\gamma ^2}}}{{{\gamma ^2} + {\sigma ^2}}}{\bf{W}}_k^H{\left( {{\bf{H}}_{{\rm{eff}}}^k} \right)^H}
 = \frac{{{\gamma ^2}}}{{{\gamma ^2} + {\sigma ^2}}}{\left( {{\bf{V}}_{{\rm{ini}}}^k} \right)\!^H},
\end{align}
where the approximation (a) is due to the asymptotic orthogonality  ${\bf M}_k^H{\bf M}_k \approx {\bf I}$ when $N_r$ is large \cite{{8}}. Given $\mu  = \frac{{{\gamma ^2}}}{{{\gamma ^2} + {\sigma ^2}}}$ and
$\forall k$,
\begin{align}
\label{equ:lemma 21}
{\bf{V}}_k^H{\bf{H}}_{{\rm{eff}}}^k{{\bf{W}}_k} &\approx \mu{\left( {{\bf{V}}_{{\rm{ini}}}^k} \right)\!^H}{\bf{H}}_{{\rm{eff}}}^k{\left( {{\bf{H}}_{{\rm{eff}}}^H{{\bf{H}}_{{\rm{eff}}}}} \right)^{ - 1}}{\left( {{\bf{H}}_{{\rm{eff}}}^k} \right)^H}{\bf{V}}_{{\rm{ini}}}^k\notag\\
 &= \mu{\left( {{\bf{V}}_{{\rm{ini}}}^k} \right)\!^H}{\bf{V}}_{{\rm{ini}}}^k \!=\!\mu{{\bf{I}}_{{N_s}}}.
\end{align}
Similarly, given ${1 \le m \le K}$, ${1 \le k \le K}$, and $m \ne k$, ${{\bf{V}}_m^H{\bf{H}}_{{\rm{eff}}}^m{{\bf{W}}_k} \approx {\bf{0}}}$.
Hence we can have ${{\bf{V}}^H}{{\bf{H}}_{{\rm{eff}}}}{\bf{W}}\!\approx\!\mu {{\bf{I}}_{K{N_s}}}$, which does not depend on $\{{\bf V}_{\rm ini}^k\}_{k=1}^K$ if $\{{\bf V}_{\rm ini}^k\}_{k=1}^K$ are arbitrary unitary matrices.
\subsection{Analog Precoding Design}
By substituting ${{\bf{V}}^H}{{\bf{H}}_{{\rm{eff}}}}{\bf{W}}\!\approx\!\mu {{\bf{I}}_{K{N_s}}}$ into (\ref{equ:eq3}),
the SMSE $\xi$ can be further expressed as
\begin{align}
\label{equ:lemma 23}
\!\!\!\xi  &\approx \left( {{\mu ^2} - 2\mu  + 1} \right)K{N_s} + \frac{{{\sigma ^2}}}{{{\gamma ^2}}}{\rm{tr}}\left( {{{\bf{V}}^H}{{\bf{M}}^H}{\bf{MV}}} \right)\notag\\
 &\mathop=\limits^b \!\left( {{\mu ^2} \!-\! 2\mu  \!+\! 1} \right)K{N_s} \!+\! \frac{{{\sigma ^2}}}{{{P_t}}}{\rm{tr}}\left( {{{\bf{F}}}{\bf{W}}{{\bf{W}}^H}{\bf{F}}^H} \right){\rm{tr}}\left( {{{\bf{V}}^H}{{\bf{M}}^H}{\bf{MV}}} \right)\notag\\
&\mathop \approx \limits^c\!\left( {{\mu ^2}\! -\! 2\mu  \!+ \!1} \right)K{N_s} \!+ \! \frac{{K{N_s}{\mu}^2{\sigma ^2}}}{{{P_t}}}{\rm{tr}}\left( {{{\left( {{\bf{H}}_{{\rm{eff}}}^H{{\bf{H}}_{{\rm{eff}}}}} \right)}^{ \!- 1}}} \right),
\end{align}
where the equation (b) is due to the power constraint in (\ref{equ:digital1}) and the approximation (c) is due to ${{\bf{V}}_k}{\bf{V}}_k^H \approx {\mu}^2{{\bf{I}}_{{N_s}}}$ and ${\bf{V}}{{\bf{V}}^H} \approx {\mu}^2{{\bf{I}}_{K{N_s}}}$ based on (\ref{equ:lemma 20}). Note that ${{\bf{F}}^H}{\bf{F}} \approx {\bf{I}}$ (${{\bf{M}}^H}{\bf{M}} \approx {\bf{I}}$) when $N_t$ ($N_r$) is large \cite{{8}}. 

In (\ref{equ:lemma 23}), when signal-to-noise ratio (SNR) is high with ${\sigma ^2}  \to 0$, we can obtain $\mu  \to 1$ and $\xi  \approx \frac{{K{N_s}{\sigma ^2}}}{{{P_t}}}{\rm{tr}}\left( {{{\left( {{\bf{H}}_{{\rm{eff}}}^H{{\bf{H}}_{{\rm{eff}}}}} \right)}^{ - 1}}} \right)$. Note that when SNR is low where the noise dominates the performance, i.e., ${\sigma ^2} \to \infty$, and then we have $\mu  \to 0$ and $\xi  \approx K{N_S}$, which is a constant. Hence, we consider the analog precoding design at high SNR. Moreover, we consider the singular value decomposition (SVD) of ${\bf H}_{\rm eff}$, i.e., ${{\bf{H}}_{{\rm{eff}}}} = {\bf{U\Sigma }}{{\bf{D}}^H}$, where ${\bf{\Sigma }} = {\rm{diag}}\left\{ {{\sigma _1},{\sigma _2}, \cdots ,{\sigma _{K{N_s}}}} \right\}$ with the descent order, and then $\xi$ can be rewritten as
\begin{align}
\xi  \approx \frac{{{KN_s}{\sigma ^2}}}{{{P_t}}}{\rm{tr}}\left( {{{\bf{\Sigma }}^{ - 2}}} \right)
 = \frac{{{KN_s}{\sigma ^2}}}{{{P_t}}}\sum\nolimits_{i = 1}^{K{N_s}} {\sigma _i^{ - 2}}.
\end{align}
To further optimize the SMSE, the analog precoding design can be formulated as the following optimization problem
\begin{align}
\label{equ:problem1}
&\mathop {\min }\limits_{\left\{ {{{\bf{M}}_k}} \right\}_{k = 1}^K,{\bf{F}}} \sum\nolimits_{i = 1}^{K{N_s}} {\sigma _i^{ - 2}} \notag\\
&{\rm{s}}.{\rm{t}}.\ {\rm{angle}}\left( {{{\left( {{{\bf{M}}_k}} \right)}_{i,j}}} \right) \in {{\cal Q}_r},{\rm{angle}}\left( {{{\bf{F}}_{i,j}}} \right) \in {{\cal Q}_t},\forall i,j,k,
\end{align}
\begin{algorithm}[b]
\caption{Proposed OSC Design}
\label{alg:alg0}
\begin{algorithmic}
\renewcommand{\algorithmicrequire}{\textbf{Input:}}
\renewcommand\algorithmicensure {\textbf{Output:} }
\REQUIRE {Number of antennas ${N=N_yN_z}$, quantization bits of phase shifters ${q}$, and the over-sampling factor $\rho$.}
\ENSURE{Predefined codebook ${\cal{D}}$.}
\\${\kern -16pt}$ \text{1.}~{Obtain the candidate set of practical quantized phases ${\cal Q} = \left\{ {0,\frac{{2\pi }}{{{2^q}}}, \cdots ,\frac{{2\pi ({2^q} - 1)}}{{{2^q}}}} \right\}$ and the sets of over-sampling spatial radian-frequencies associated with the azimuth and elevation angles ${{\cal R}_y} = \left\{ {0,\frac{{2\pi }}{{\rho N_y }}, \cdots ,\frac{{2\pi (\rho N_y - 1)}}{{\rho N_y }}} \right\}$, ${{\cal R}_z} = \left\{ {0,\frac{{2\pi }}{{\rho N_z }}, \cdots ,\frac{{2\pi ( \rho N_z - 1)}}{{\rho N_z }}} \right\}$.}
\\${\kern -14pt}$ \text{2.}~{Obtain the codebook ${\cal D}\!\! = \!\!\left\{ {{\bf{a}}\left( {{w_y},{w_z}} \right)|{w_y} \in {{\cal R}_y},{w_z}\in {{\cal R}_z}} \right\}$.}
\\${\kern -22pt}$ \text{3.}~{Quantize phases of all candidates in ${\cal D}$ as ${\mathop{\rm angle}\nolimits} ({\bf{a}}{({w_y},\!{w_z})_{i,j}})\!\! =\!\! \mathop {\arg \min }\limits_{\alpha  \in {\cal Q}} \left\| {\rm angle}{({{\bf{a}}{{\left( {{w_y},\!{w_z}} \right)}_{i,j}})}\! -\! \alpha } \right\|_2^2$, where ${{\bf{a}}\left( {{w_y},{w_z}} \right)}$ is one specific candidate in ${\cal D}$, $1 \le i \le N$, and $j=1$.}
\\${\kern -14pt}$ \text{4.}~{${\cal D}={\rm unique}( \cal D)$ to make candidates in $\cal D$ be unique.}
\end{algorithmic}
\end{algorithm}
\!\!where ${\cal Q}_r$ and ${\cal Q}_t$ are the candidate sets of quantized phases due to the finite resolution of the phase shifters. From (\ref{equ:problem1}), we observe that given $\sum\nolimits_{i = 1}^{K{N_s}} {\sigma _i^{ - 2}}$, ${\bf H}_{\rm eff}$ with equal singular values will reach the minimization of (\ref{equ:problem1}). This enlightens us to design the analog precoder/combiner with the condition number ${\rm{cond}}\left( {\bf H}_{\rm eff} \right)$ as small as possible, where ${\rm{cond}}\left( {\bf H}_{\rm eff} \right) = {\sigma _1}/{\sigma _{KN_s}}$. On the other hand, given the equal singular values, the minimization of (\ref{equ:problem1}) is equivalent to
\begin{align}
\label{equ:problem2}
&\mathop {\max }\limits_{\left\{ {{{\bf{M}}_k}} \right\}_{k = 1}^K,{\bf{F}}} \sum\nolimits_{i = 1}^{K{N_s}} {\sigma _i^{2}} \notag\\
&{\rm{s}}.{\rm{t}}.\ {\rm{angle}}\left( {{{\left( {{{\bf{M}}_k}} \right)}_{i,j}}} \right) \in {{\cal Q}_r},{\rm{angle}}\left( {{{\bf{F}}_{i,j}}} \right) \in {{\cal Q}_t},\forall i,j,k.
\end{align}
Moreover, even the optimization problem (\ref{equ:problem2}) can also be challenging, which requires the prohibitively high complexity with ${\cal O}\left( {{2^{K{B_r}{M_r}{N_r}{B_t}{M_t}{N_t}}}} \right)$, where $B_r$ and $B_t$ are the quantization
bits of phase shifters at the users and BS, respectively.
While if ${\bf H}_{\rm eff}$ is a diagonally-dominant matrix, the optimization problem (\ref{equ:problem2}) can be approximated to $\mathop {{\rm{max}}}\limits_{{\{ {{{\bf{M}}_k}} \}_{k = 1}^K}\!,{\bf{F}}} \sum\nolimits_{i = 1}^{K{N_s}} \!\!{\left| \left({{\bf{H}}_{{\rm{eff}}}} \right)_{i,i}\right|\!^2}$.
Based on the analysis above, we consider the following multi-object optimization problem
\begin{align}
\label{equ:problem3}
\!\!\!&\mathop {{\rm{max}}}\limits_{\left\{ {{{\bf{M}}_k}} \right\}_{k = 1}^K,{\bf{F}}} {\bf{g}}\left( {\left\{ {{{\bf{M}}_k}} \right\}_{k = 1}^K,{\bf{F}}} \right) \!\!=\!\! {\left[ {\sum\limits_{i = 1}^{K{N_s}} {{{\left| \left({{\bf{H}}_{{\rm{eff}}}}\right)_{i,i} \right|}}}\!^2 , - {\rm{cond}}\left( {{{\bf{H}}_{{\rm{eff}}}}} \right)} \right]^T}\notag\\
&{\rm s.t.}\  {\bf{M}}_k^{\left (m \right)}  \in {\cal D}_{k,r}, {\bf{F}}^{\left (m \right)}  \in {\cal D}_t, \forall m,k,
\end{align}
\begin{algorithm}[b]
\caption{Proposed JAPC Design}
\label{alg:alg1}
\begin{algorithmic}
\renewcommand{\algorithmicrequire}{\textbf{Input:}}
\renewcommand\algorithmicensure {\textbf{Output:} }
\REQUIRE {Maximum correlation factor $\beta$, BS codebook ${\cal{D}}_t$, user codebook ${\cal{D}}_{k,r}$, and channel ${\bf H}_k$ for ${1 \le k \le K}$.}
\ENSURE{Analog precoder ${\bf F}$ and analog combiner ${\bf M}_k$, $\forall k$.}
\\${\kern -14pt}$ \text{1.}~{Initialization: ${\cal K} = \{ 1,2, \cdots ,K\}$, ${\bf{F}} = {{\bf{M}}_k} = \Phi, \forall k$, and $\Phi$ denotes empty matrix;}
\\${\kern -14pt}$ \text{2.}~{\textbf{for} $i_{\rm iter} = 1:K{M_r}$}
\\${\kern -14pt}$ \text{3.}~~~{$\{ {k^ * },{\bf{a}}_r^ * ,{\bf{a}}_t^ * \}  = \arg \mathop {\max }\limits_{k \in {\cal K},{{\bf{a}}_r} \in {{\cal D}_{k,r}},{{\bf{a}}_t} \in {{\cal D}_t}} \left\| {{\bf{a}}_r^H{{\bf{H}}_k}{{\bf{a}}_t}} \right\|_2^2$;}
\\${\kern -14pt}$ \text{4.}~~~{${{\bf{M}}_{k^*}} = [{{\bf{M}}_{k^*}}|{\bf{a}}_r^ * ]$, ${{\bf{F}}_{k^*}} = [{{\bf{F}}_{k^*}}|{\bf{a}}_t^ * ]$;}
\\${\kern -14pt}$ \text{5.}~~~{${{\cal D}_{{k^*}\!,r}} \!= \!{\rm setdiff}\left({{\cal D}_{{k^*},r}}, \left\{{\bf{a}}_r| {{\bf{a}}_r} \!\in\! {{\cal D}_{{k^*},r}}, \big| {{{\left( {{\bf{a}}_r^ * } \right)}\!^H}{{\bf{a}}_r}} \big| \!\ge \! \beta \right\} \right)$;}
\\${\kern -14pt}$ \text{6.}~~~{${{\cal D}_{t}} \!= \!{\rm setdiff}\left({{\cal D}_{t}}, \left\{ {{\bf{a}}_t}|{{\bf{a}}_t} \in {{\cal D}_{t}}, \big| {{{\left( {{\bf{a}}_t^ * } \right)}^H}{{\bf{a}}_t}} \big| \!\ge\! \beta \right\} \right)$;}
\\${\kern -14pt}$ \text{7.}~~~{${\cal K}= {\rm setdiff}\left({\cal K}, \!\left\{ {k^ *} \right\}\right)$ when the number of columns of \\~~${\bf{M}}_{k^*}$ equals $M_r$;}
\\${\kern -14pt}$ \text{8.}~{\textbf{end for}}
\\${\kern -14pt}$ \text{9.}~{${\bf{F}} = [{{\bf{F}}_1},{{\bf{F}}_2}, \cdots ,{{\bf{F}}_K}]$.}
\end{algorithmic}
\end{algorithm}
\!\!where ${\cal D}_{k,r}$ and ${\cal D}_{t}$ are the predefined codebooks for ${\bf M}_k$ and ${\bf F}$, respectively.
Thus, the computational complexity to acquire the solution to (\ref{equ:problem3}) can be reduced to ${\cal O}\left(K {{{\left| {\cal D}_{k,r} \right|}_c}} {{{\left| {\cal D}_t \right|}_c}}\right)$, and ${{\left| {\cdot} \right|}_c}$ denotes the cardinality of a set.
To ensure ${\bf H}_{\rm eff}$ to be a diagonally-dominant matrix, we propose that each candidate from ${\cal D}_{k,r}$ and ${\cal D}_t$ has the form of UPA response vector as shown in (\ref{equ:model1}), i.e.,
\begin{align}
\label{equ:vec}
{\bf{a}}\left( {{w_y},{w_z}} \right) &= \frac{1}{{\sqrt {N} }}\big[1, \cdots ,{e^{j\left( {n{w_y} + m{w_z}} \right)}}, \cdots , \notag \\&\
{e^{j\left( {\left( {N_{y} - 1} \right){w_y} + \left( {N_{z} - 1} \right){w_z}} \right)}}{\big]^T},
\end{align}
where ${N_y N_z = N}$, ${w_y}$ and ${w_z}$ denote the spatial frequencies associated with the azimuth and elevation angles, respectively.

The conventional 2D-DFT codebook has the limited spatial resolution. To further improve the spatial resolution, we design an OSC as shown in \textbf{Algorithm \ref{alg:alg0}}, where the set ${\rm unique}({\cal D})$ has the same candidates as the set ${\cal D}$ without repetitions.
Specifically, we first obtain an OSC ${\cal D}$ given the over-sampling factor ${\rho}$ (steps 1$\sim$2). For each candidate in ${\cal D}$, we quantize the phases of all elements by considering the quantization bits of phase shifters ${q}$ (step 3), and finally make candidates in ${\cal D}$ be unique (step 4). Conventional 2D-DFT codebook is a special case of the proposed OSC for UPA when ${\rho} = 1$. When ${\rho}$ becomes large, the proposed codebook can achieve the better spatial resolution as illustrated in Fig. \ref{fig:fig5}. On the other hand, the better spatial resolution indicates the stronger correlation among candidates in ${\cal D}$. This implies that directly acquiring analog precoder/combiner ${{\bf M}_k}$ and ${{\bf F}_k}$ to maximize $\sum\nolimits_{i = 1}^{K{N_s}} {{\left| {\left({\bf{H}}_{{\rm{eff}}}\right)_{i,i}} \right|}^2}$ will lead ${{\bf H}^k_{\rm eff}}$ to be ill-conditioned or even rank-insufficient.

To solve the optimization problem (\ref{equ:problem3}), we further propose a heuristic JAPC algorithm as shown in \textbf{Algorithm \ref{alg:alg1}}, where the set ${\rm setdiff}({\cal A}, {\cal B})$ has all candidates in ${\cal A}$ that are not in ${\cal B}$.
Specifically, in each iteration, we first
acquire the user index $k^*$ and the associated \{${\bf a}_r^*$, ${\bf a}_t^*$\} that maximize the diagonal elements of effective channels (step 3) to be part of the analog precoder/combiner matrices ${\bf M}_k$ and ${\bf F}_k$, respectively (step 4).
To prevent ${\bf H}_{\rm eff}$ to be ill-conditioned, the candidates in ${\cal D}_{k,r}$ and ${\cal D}_t$ having strong correlation with the selected \{${\bf a}_r^*$, ${\bf a}_t^*$\} in \textbf{step 3} are removed (steps 5$\sim$6). Finally, if the $k^*$-th user's analog combiner matrix  ${\bf M}_k^*$ is completed, the index $k^*$ is removed from ${\cal K}$ (step 7).
\begin{figure}
     \centering
     \includegraphics[width=8.5cm, keepaspectratio]%
     {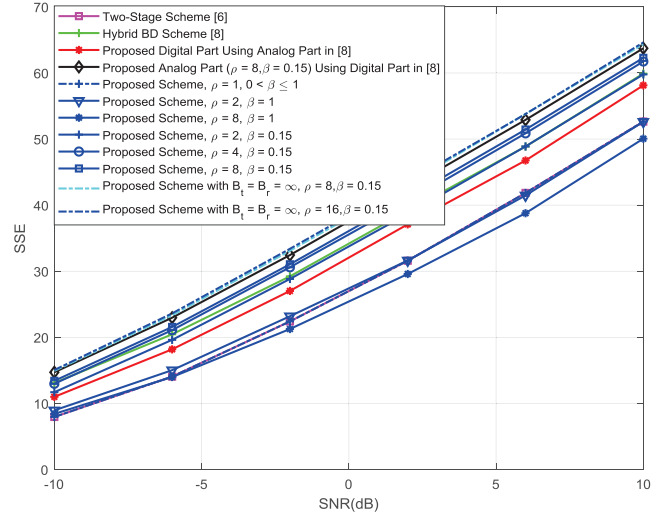}
     \caption{SSE achieved by different schemes in a multi-user MIMO system, where ${K = 8}$, ${M_t} = 8$, ${M_r} = {N_s} = 1$.}
     \label{fig:s1}
\end{figure}
\begin{figure}
     \centering
     \includegraphics[width=8.8cm, keepaspectratio]%
     {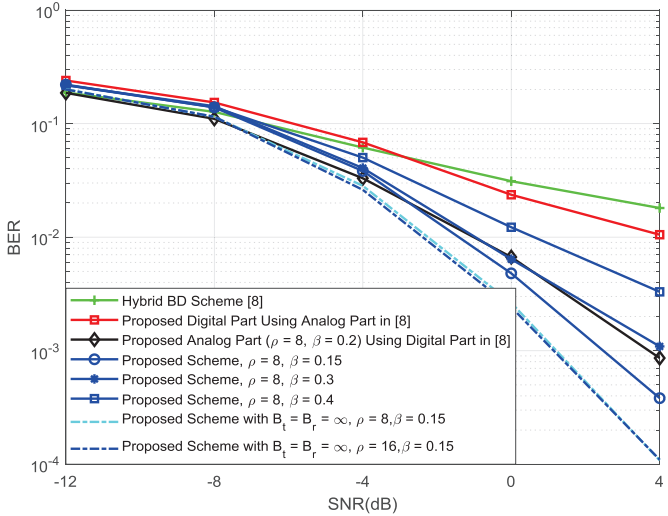}
     \caption{BER achieved by different schemes in a multi-user MIMO system, where ${K = 2}$, ${M_t} = 4$, ${M_r} = {N_s} = 2$.}
     \label{fig:b1}
\end{figure}
\section{Simulation Results}
The simulation parameters are shown as follows: BS and users are equipped with $8 \times 8$ and $4 \times 4$ UPAs, respectively, ${B_t} = 3$ bits, ${B_r} = 2$ bits, and we consider the 16-QAM modulation.
For the mmWave channel model, we consider ${N_c} = 8$ and ${N_p} = 10$,
azimuth/elevation AoAs and AoDs follow the uniform distribution ${\cal U}\left[ { - \pi / 2,\pi  / 2} \right]$,
 and all the angle spreads are ${7.5^ \circ }$ \cite{{6},{8}}. We will investigate the performance of the proposed scheme by comparing it with the two-stage hybrid precoding scheme \cite{{4}}, the hybrid BD scheme \cite{{6}}, the proposed digital part using analog part in \cite{{6}}, and our proposed analog part using digital part (i.e., BD precoding) in \cite{{6}}. Besides, the proposed scheme with ${B_t = B_r = \infty }$ (i.e., the step 3 and step 4 in Algorithm 1 are removed) are simulated to investigate the performance gap compared to the optimal schemes.


Fig. \ref{fig:s1} and Fig. \ref{fig:b1} investigate the SSE and BER performance, respectively.
In Fig. \ref{fig:s1}, the proposed scheme with $\rho = 8$ and $\beta = 0.15$ has the better SSE than the two-stage scheme \cite{{4}} and the hybrid BD scheme \cite{{6}}. Especially, the proposed OSC-based JAPC design with the digital part in \cite{{6}} outperforms the hybrid BD scheme \cite{{6}}.
Besides, the SSE of the proposed scheme with $\beta = 1$ becomes worse when $\rho$ increases, since directly using solution to (\ref{equ:problem3}) with OSC will lead large ${\rm{cond}}\left( {{{\bf{H}}_{{\rm{eff}}}}} \right)$. While the SSE of the proposed scheme with $\beta = 0.15$ becomes better when $\rho$ increases, since larger $\rho$ can achieve better spatial resolution and the proposed JAPC design with a suitable $\beta$ can guarantee small ${\rm{cond}}\left( {{{\bf{H}}_{{\rm{eff}}}}} \right)$.
Note that when $\rho = 1$, the OSC is reduced to conventional 2D-DFT codebook, and the JAPC design with $\rho = 1$ for $0 < \beta \le 1$ is equivalent since elements in 2D-DFT codebook are orthogonal.
In Fig. \ref{fig:b1}, we observe that when BER $=10^{-2}$ is considered, the proposed scheme ($\rho = 8$ and $\beta = 0.15$), the proposed min-SMSE based digital part using analog part in \cite{{6}}, and the proposed OSC-based JAPC design using analog part in \cite{{6}} have over 5.1 dB, 1 dB, and 5 dB SNR gains than the hybrid BD scheme, respectively.
Note that when ${B_t = B_r  = \infty}$, the proposed scheme with $\rho = 8$ and the proposed scheme with $\rho = 16$ have very similar BER and SSE performance, which indicates OSC with $\rho = 8/16$ can have sufficiently large spatial resolution and these two curves can be considered as the optimal performance with infinite spatial resolution of codebook and infinite quantization of phase shifters. The BER and SSE performance of the proposed scheme with $\rho = 8$ and $\beta = 0.15$ still has a gap compared to the optimal schemes due to the limited quantization of phase shifters.
\section{Conclusions}
This letter has proposed an OSC-based hybrid min-SMSE precoding scheme to support multi-user and multi-stream transmission for each user in mmWave 3D-MIMO systems.
Specifically, we design the digital precoder/combiner to improve the BER performance based on the min-SMSE criterion. Moreover, in the analog part, we propose an OSC and the associated JAPC design by further reducing the SMSE.
Finally, simulation results show the proposed scheme can achieve the better performance than the conventional schemes.


\end{document}